\documentclass[a4paper]{llncs}
\usepackage{amssymb}
\usepackage{graphicx}

\usepackage{url}
\urldef{\mailsa}\path|verica.bakeva@finki.ukim.mk, aleksandra.popovska.mitrovikj@finki.ukim.mk,|
\urldef{\mailsb}\path|vesna.dimitrova@finki.ukim.mk|
\newcommand{\keywords}[1]{\par\addvspace\baselineskip
\noindent\keywordname\enspace\ignorespaces#1}

\begin{document}

\mainmatter  

\title{Resistance of Statistical Attacks of Parastrophic Quasigroup Transformation}

\titlerunning{Resistance of Statistical Attacks of Parastrophic Quasigroup Transformation}

%
%
\author{Verica Bakeva \and Aleksandra Popovska-Mitrovikj \and Vesna Dimitrova}
\authorrunning{V. Bakeva A. Popovska-Mitrovikj  and  V. Dimitrova}


%

\institute{Faculty of Computer Science and Engineering,\\
Ss. Cyril and Methodius University, Skopje, Macedonia
\mailsa\\
\mailsb\\
\url{}}

%
%

\toctitle{Resistance of Statistical Attacks of Parastrophic
Quasigroup Transformation} \maketitle

\begin{abstract}
In this paper, we prove an important cryptographic property of $PE$-transformation proposed elsewhere.
If $PE$-transformation is used as encrypting function then after $n$ applications of it on arbitrary message
the distribution of $l$-tuples ($l=1,2,\dots, n$) is uniform. This property implies the resistance of
statistical kind of attack of this transformation. For illustration of theoretical results,
some experimental results are presented as well.

\keywords{cryptographic properties, statistical attack, encrypted message, quasigroup, parastrophic quasigroup transformation,
uniformity}


AMS Mathematics Subject Classification (2010): 94A60, 20N05, 60J20
\end{abstract}

\section{Introduction}

Quasigroups and quasigroup transformations are very useful for
construction of cryptographic primitives, error
detecting and error correcting codes. The reasons for that are  the structure of quasigroups, their large
number, the properties of quasigroup transformations and so on.
The quasigroup string transformations $E$ and their properties
were considered in several papers.

A quasigroup $(Q,*)$ is a groupoid (i.e. algebra with
one binary operation $*$ on the finite set $Q$) satisfying the
law:
\begin{equation}
(\forall u,v\in Q)(\exists !x,y\in Q)\ \ (x*u=v \ \&\ u*y=v)
\label{eden}
\end{equation}

In fact, (\ref{eden}) says that a groupoid $(Q,*)$ is a quasigroup
if and only if the equations $x*u=v$ and $u*y=v$ have unique
solutions $x$ and $y$ for each given $u,v \in Q$.

In the sequel, let $A=\{1,\dots, a\}$ be an alphabet of integers ($a\geq 2$) and denote by
$A^+=\{x_1\dots x_k |$ $\  x_i \in A, \ k\ge 1\}$ the set of all
finite strings over $A$. Note that $A^+=\displaystyle\bigcup_{k\ge
1} A^k$, where $A^k=\{x_1\dots x_k |\ x_i \in A\}$. Assuming that
$(A,*)$ is a given quasigroup, for any letter $l\in A$ (called
leader), Markovski and al. (see \cite{SDV-1}) defined the transformation  $E=E^{(1)}_l:A^+ \rightarrow A^+$
by
\begin{equation}
\label{funkcija_E} E(x_1\dots x_k)=y_1\dots y_k \Leftrightarrow
\left\{
\begin{array}{cll}y_1&=&l*x_1,\\
y_{i}&=&y_{i-1} * x_{i},\quad i=2,\ldots,k
\end{array}\right.
\end{equation}
where $x_i,y_i \in A$. Then, for given quasigroup operations $*_1,
*_2, \dots, *_n$ on the set $A$,  we can define mappings $E_1,
E_2, \dots, E_n$, in the same manner as previous by choosing fixed
elements $l_1, l_2, \dots, l_n \in A$ (such that $E_i$ is
corresponding to $*_i$ and $l_i$). Let
$$E^{(n)} = E^{(n)}_{l_n,\dots,l_1} = E_n\circ E_{n-1}\circ \dots\circ E_1,$$
where $\circ$ is the usual composition of mappings ($n\ge 1$). It is easy to
check that  the mappings $E$ is a bijection. In the same paper, authors proposed a transformation  $E^{(n)}$ as an encryption function and proved the following theorem.
\begin{theorem}
Let $\alpha\in A^+$  be an arbitrary string and
$\beta=E^{(n)}(\alpha)$. Then $m$-tuples in $\beta$ are uniformly
distributed for $m \le n$. \label{Teorema stara}
\end{theorem}

Also, in Theorem 2 in \cite{VV}, Bakeva and Dimitrova proved that the  probabilities of $(n+1)$-tuples in $\beta=E^{(n)}(\alpha)$ are divided in $a$ classes where $a=|A|$,  if $(p_1,p_2,\dots, p_a)$ is the distribution of letters in an input
string and  $p_1,p_2,\dots, p_a$ are distinct probabilities, i.e., $p_i\ne p_j$ for $i\ne j$. Each  class  contains  $a^n$ elements with  the same  probabilities  and  the  probability  of  each $(n+1)$-tuple in $i$-th class is
$\displaystyle \frac{1}{a^n}p_i$, for $i=1,2,\dots, a$. If
$p_{i_1}=p_{i_2}=\dots=p_{i_\nu}$ for some $1\le
i_1<\dots<i_\nu\le a$, then the classes
with probabilities $\displaystyle \frac{1}{a^n}p_{i_1}= \frac{1}{a^n}p_{i_2}=\dots=\displaystyle
\frac{1}{a^n}p_{i_\nu}$ will be merged in one class with $\nu a^n$
elements. Using these results, the authors proposed an algorithm for cryptanalysis.

In paper \cite{Krapez}, Krapez gave an idea for a new quasigroup string
transformation based on parastrophes of quasigroups. A modification of this quasigroup transformation is defined in \cite{VVA}.  In \cite{VVAA}, authors showed that the parastrophic quasigroup transformation has good properties for application in  cryptography. Namely, using that transformation the number of quasigroups of order 4 useful in cryptography is increased. To complete the proof of goodness of parastrophic quasigroup transformation for cryptography, it is needed to prove that Theorem \ref{Teorema stara} holds for that transformation, too. It will guarantee that message encrypted by the parastrophic quasigroup transformation will be resistant of a statistical kind of attacks.

\section{Parastrophic transformation}

In this Section, we briefly repeat the construction of parastrophic quasigroup transformation given in \cite{VVA}.

Recall that every quasigroup $(Q,*)$ has a set of five quasigroups,
called parastrophes denoted with $/, \backslash, \cdot,$ $
//, \backslash \backslash$  which are defined in Table 1.\\

\begin{center}
\begin{table}[h]
\caption{ Parastrophes of quasigroup operations $\ast$}
\begin{center}
$\begin{array}{lllll} \hline \multicolumn {5}{c}{{\rm
Parastrophe\ operations}}\\
\hline
x\backslash y=z&\Longleftrightarrow &x\ast z = y&& \\
x/y=z&\Longleftrightarrow &z\ast y=x&& \\
x\cdot y=z&\Longleftrightarrow & y\ast x = z&&\\
x//y=z&\Longleftrightarrow& y/x=z&\Longleftrightarrow &z \ast x=y\\
x\backslash \backslash y=z&\Longleftrightarrow& y\backslash
x=z&\Longleftrightarrow &y \ast z =x\\\hline
\end{array}$
\end{center}
\end{table}
\end{center}

In this paper we use the following notations for parastrophe
operations:
$$\begin{array}{lcl}
f_1(x,y)=x*y, & f_2(x,y)=x\backslash y, & f_3(x,y)=x/y,\\
f_4(x,y)=x\cdot y,&f_5(x,y)=x//y,& f_6(x,y)=x\backslash \backslash
y.
\end{array}$$

Let $M=x_1 x_2 \dots x_k$ be an input message. Let  $d_1$  be an random integer such that $(2\le d_1< k)$ and $l$ be random chosen element (leader) from $A$. Also, let $(A,*)$ be a quasigroup and $f_1$, \dots, $f_6$ be its parastrophe operations.

Using previous transformation $E$, for chosen $l$, $d_1$ and quasigroup $(A,*)$ we define a
parastrophic transformation $PE=PE_{l,d_1}:A^+ \rightarrow A^+$ as
follows.

At first, let  $q_1=d_1$ be the length of the first block, i.e., $M_1=x_1x_2\dots x_{q_1}$. Let $s_1=(d_1 \ \mbox{\rm mod } 6) +1$. Applying the transformation $E$ on the block $M_1$ with leader $l$ and quasigroup operation $f_{s_1}$, we obtain the encrypted block
$$C_1=y_1y_2\dots y_{q_1-1}y_{q_1}=E_{f_{s_1},l}(x_1x_2\dots x_{q_1-1}x_{q_1}).$$
Further   on,   using  last  two   symbols  in  $C_1$ we calculate
the   number   $d_2=4y_{q_1-1}+y_{q_1}$ which determines  the
length  of  the  next  block.  Let $q_2=q_1+d_2$,
$s_2=(d_2 \ \mbox{\rm mod }  6) +1$ \ \ and\ \ $M_2=x_{q_1+1}\dots
x_{q_2-1}x_{q_2}$. After applying $E_{f_{s_2},y_{q_1}}$, the
encrypted block $C_2$ is

$$C_2=y_{q_1+1}\dots y_{q_2-1}y_{q_2}=E_{f_{s_2},y_{q_1}}(x_{q_1+1}\dots x_{q_2-1}x_{q_2}).$$

In general case, for given $i$, let the encrypted blocks
$C_1$,\dots, $C_{i-1}$ be obtained and $d_i$ be calculated using
the last two symbols in $C_{i-1}$, i.e., $d_i=4y_{q_i-1}+y_{q_i}$. Let
$q_i=q_{i-1}+d_i$, $s_i=(d_i  \ \mbox{\rm mod }  6)+1$ and
$M_i=x_{q_{i-1}+1}\dots x_{q_i-1}x_{q_i}$. We apply the
transformation $E_{f_{s_i},y_{q_{i-1}}}$ on the block $M_i$ and
obtain the encrypted block
$$C_i=E_{f_{s_i},y_{q_{i-1}}}(x_{q_{i-1}+1}\dots x_{q_i}).$$
Now, the parastrophic transformation is defined as
\begin{equation}
PE_{l,d_1}(M)=PE_{l,d_1}(x_1 x_2 \dots x_n)=C_1||C_2||\dots||C_r,%
\label{PE_transformacija}
\end{equation}
where $||$ is a concatenation of blocks.
Note that the
length of the last block $M_r$ may be shorter than $d_r$ (depends
on the number of letters in the input message). The transformation
$PE$ is schematically presented in Figure 1.

\begin{figure}[h!]
\includegraphics[width=3.7in]{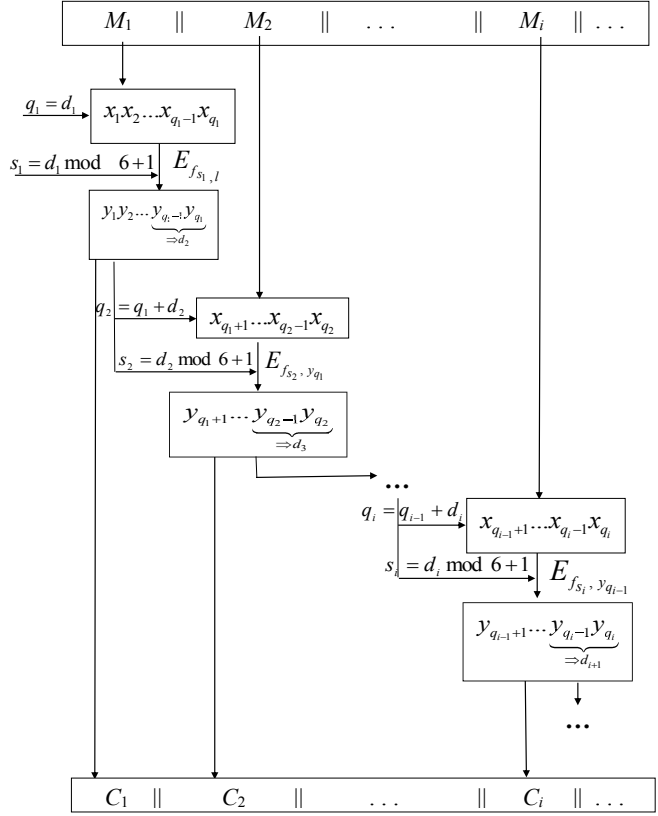}%
\vspace{-0.8 cm}

\caption{Parastrophic transformation $PE$}
\end{figure}

For arbitrary quasigroup on a set $A$, random leaders $l_1,\dots l_n$ and random lengths $d_1^{(1)},\dots, d_1^{(n)}$, we
define mappings $PE_1$, $PE_2$, \dots, $PE_n$  as in
(\ref{PE_transformacija}) such that $PE_i$ is corresponding to $d_1^{(i)}$ and $l_i$. Using them, we define the transformation $PE^{(n)}$ as follows:
$$PE^{(n)} = PE^{(n)}_{(l_n,d_1^{(n)}),\dots,(l_1,d_1^{(1)})} = PE_n\circ PE_{n-1}\circ \dots\circ PE_1,$$
where $\circ$ is the usual composition of mappings.

\section{Theoretical proof for resistance of statistical kind of attacks}

Let the alphabet $A$ be  as  above.  A  randomly
chosen element  of  the  set  $A^k$  can be considered as a
random vector $(X_1,X_2,\dots, X_k)$, where $A$ is the range of
$X_i$, $i=1,\dots,k$. We consider these vectors as input messages.
The transformation
$PE=PE_{l,d_1}:A^{+}\rightarrow A^{+}$ can be defined as:

{\small
\begin{equation}
\label{funkcija_PE}
\begin{array}{l}
PE_{l,d_1}(X_1,\dots, X_k)=(Y_1,\dots, Y_k)\Leftrightarrow\\ \\
\qquad\qquad\Leftrightarrow \left\{
\begin{array}{l}
Y_1=f_{s_{1}}(l,X_1), \ \ Y_{j}=f_{s_{1}}(Y_{j-1},X_{j}),\  j=2,\ldots,d_1,\\
\\
Y_{q_{i}+j}=f_{s_{i+1}}(Y_{q_{i}+j-1},X_{q_{i}+j}), i=1,\ldots,r-1,\ j=1,\ldots,d_{i+1}
\end{array}\right.
\end{array}
\end{equation}}

Let $(p_{1},p_{2},..., p_{a})$ be the probability distribution
of the letters $1,...,a$ in an input message. That implies
$p_{i}>0$  for each $i=1,2,...,a$ and $\displaystyle\sum_{i=1}^{a} p_i=1$.

An important property of one transformation for application in cryptography is the uniform distribution of the substrings in the output message $(Y_1,\dots, Y_k)$.
This property guarantee the resistance of statistical attack. Therefore, we investigate the distribution of substrings in the output message
obtained using $PE$-transformation. At first we will prove that after applying the transformation $PE^{(1)}$ on an input message $\alpha$,
the letters in transformed message are uniformly distributed.

\begin{theorem}
\label{uniformity of letters}
The letter $Y_t$ has uniform distribution on the set $A=\{1,\dots, a\}$, i.e., $Y_t\sim U(\{1,\dots, a\})$ for each $t$ ($t=1, 2,\dots, k$).
\end{theorem}

\begin{proof}{In this proof we use the same notations as in construction of parastrophic quasigroup transformation given in the previous section.

At first, note that the leader $l$ can be consider as uniformly distributed random variables on the set $A$ since it is randomly chosen from the set $A$. Therefore, $l\sim U(\{1,\dots, a\})$, i.e.,
$$P\{l=i\}=\frac{1}{a}, \qquad \mbox{\rm for each } i\in A.$$
Also, leader $l$ is independent of each letter $X_i$ in the input message.

Let $t=1$. Using the equation (\ref{funkcija_PE}) and  total probability theorem, for distribution of $Y_1$, we obtain
$$\begin{array}{lcl}
P\{Y_1=j\}&=&P\{f_{s_1}(l,X_1)=j\}\\
&=&\displaystyle \sum_{i=1}^a P\{l=i\}P\{f_{s_1}(l,X_1)=j | l=i\}\\
&=&\displaystyle \sum_{i=1}^a \frac{1}{a}P\{f_{s_1}(l,X_1)=j | l=i\}\\
&=&\displaystyle \sum_{i=1}^a \frac{1}{a}P\{f_{s_1}(i,X_1)=j\}\\
&=&\displaystyle \frac{1}{a} \sum_{i=1}^a P\{X_1=f_{s_1}^{\prime}(i,j)\}
\end{array}$$
Here, $f_{s_{1}}^{\prime}$ is the inverse quasigroup transformation of $f_{s_{1}}$, i.e. if $f_{s_{1}}(u,x)=v$, then $f_{s_{1}}^{\prime}(u,v)=x$. Note that if $i$ runs over all values of $A$ then for fixed $j$, the expression
$X_1=f_{s_1}^{\prime}(i,j)$ runs over all values of A, too. Therefore,
$$P\{Y_1=j\}=\frac{1}{a} \sum_{i=1}^a P\{X_1=f_{s_1}^{\prime}(i,j)\}=\frac{1}{a} \sum_{i=1}^a p_i =\frac{1}{a},$$
i.e., $Y_1\sim U(\{1,\dots, a\})$.

We proceed by induction, and let suppose that $Y_r\sim U(\{1, 2,\dots,a\})$.
Similarly as previous, using that $f_{s_{r+1}}$ is the parastrophe operation applied in $(r+1)^{th}$ step we compute the distribution of $Y_{r+1}$ as follows.%
$$\begin{array}{lcl}
P\{Y_{r+1}=j\}&=&P\{f_{s_{r+1}}(Y_r, X_{r+1})=j\}\\
&=&\displaystyle \sum_{i=1}^a P\{Y_r=i\}P\{f_{s_{r+1}}(Y_r, X_{r+1})=j | Y_r=i\}\\
&=&\displaystyle \sum_{i=1}^a \frac{1}{a}P\{f_{s_{r+1}}(i, X_{r+1})=j | Y_r=i\}
\end{array}$$
According to definition of parastrophic operation given with (\ref{funkcija_PE}), we can conclude that the random variables $X_{r+1}$ and $Y_r$ are independent. Applying that in previous equation, we obtain
$$\begin{array}{lcl}
P\{Y_{r+1}=j\}&=&\displaystyle\sum_{i=1}^a \frac{1}{a}P\{f_{s_{r+1}}(i, X_{r+1})=j\}\\
&=&\displaystyle \frac{1}{a} \sum_{i=1}^a P\{X_{r+1}=f_{s_{r+1}}^{\prime}(i,j)\}\\
&=&\displaystyle \frac{1}{a}.
\end{array}$$
As previous, $f_{s_{r+1}}^{\prime}$ is the inverse quasigroup transformation of $f_{s_{r+1}}$.
In the last equation, we use that $X_{r+1}=f_{s_{r+1}}^{\prime}(i,j)$ runs over all values of A when $j$ is fixed and $i$ runs over all values of $A$, i.e.
$$\sum_{i=1}^a P\{X_{r+1}=f_{s_{r+1}}^{\prime}(i,j)\}=\sum_{i=1}^a p_i =1.$$
On this way, we proved that $Y_t$ has uniform distribution on the set $A$, for each $t\ge 1$.
}\end{proof}

From the Theorem 2 we can conclude the follows. If  $M \in A^k$ and $C=PE_{l,d_1}(M)$ then the
letters in the message $C$ are uniformly distributed, i.e., the probability of the appearance of a letter $i$ at the arbitrary place of the string $C$ is  $\displaystyle \frac{1}{a}$, for each $ i \in A$.

\begin{theorem}
Let $M \in A^+$  be an arbitrary string and
$C=PE^{(n)}(M)$. Then the $m$-tuples in $C$ are uniformly
distributed for $m \le n$.
\label{uniformity 2}
\end{theorem}
\begin{proof} Let $(Y_{1}^{(n)},Y_{2}^{(n)},\dots, Y_{k}^{(n)})=PE^{(n)}(X_1, X_2,\dots, X_k)$.
We will prove this theorem by induction. Let suppose that the
statement is satisfied for $n=r$, i.e.,
$(Y_{t+1}^{(r)},Y_{t+2}^{(r)}\dots Y_{t+l}^{(r)})\sim
U(\{1,2,\dots, a\}^l)$ for each $1\leq l\leq r$ and each $t\geq
0$. Now, let $n=r+1$. We consider the distribution of
$(Y_{t+1}^{(r+1)},Y_{t+2}^{(r+1)}\dots Y_{t+l}^{(r+1)})$ for each
$1\leq l\leq r+1$ and arbitrary $t$.

$$\begin{array}{l}
P\{Y_{t+1}^{(r+1)}=y_{t+1}^{(r+1)},Y_{t+2}^{(r+1)}=y_{t+2}^{(r+1)},\dots,Y_{t+l}^{(r+1)}=y_{t+l}^{(r+1)}\}\\
\quad =P\{Y_{t+1}^{(r+1)}=y_{t+1}^{(r+1)},f_{s_{t+2}}(Y_{t+1}^{(r+1)},Y_{t+2}^{(r)})=y_{t+2}^{(r+1)},\dots\\
\qquad\qquad\qquad\qquad\qquad\qquad\qquad\qquad\qquad  \dots, f_{s_{t+l}}(Y_{t+l-1}^{(r+1)},Y_{t+l}^{(r)})=y_{t+l}^{(r+1)}\},
\end{array}$$
where $f_{s_j}$ is the parastrophe operation applied in the step $j$ and $f_{s_j}^{\prime}$ is its inverse transformation, $j=t+2,\dots, t+l$. Now,
$$\begin{array}{l}
P\{Y_{t+1}^{(r+1)}=y_{t+1}^{(r+1)},Y_{t+2}^{(r+1)}=y_{t+2}^{(r+1)},\dots,Y_{t+l}^{(r+1)}=y_{t+l}^{(r+1)}\}\\
\quad =P\{Y_{t+1}^{(r+1)}=y_{t+1}^{(r+1)},f_{s_{t+2}}(y_{t+1}^{(r+1)},Y_{t+2}^{(r)})=y_{t+2}^{(r+1)},\dots\\
\qquad\qquad\qquad\qquad\qquad\qquad\qquad\qquad\qquad  \dots, f_{s_{t+l}}(y_{t+l-1}^{(r+1)},Y_{t+l}^{(r)})=y_{t+l}^{(r+1)}\}\\
\quad =P\{Y_{t+1}^{(r+1)}=y_{t+1}^{(r+1)}, Y_{t+2}^{(r)}=f_{s_{t+2}}^{\prime}(y_{t+1}^{(r+1)},y_{t+2}^{(r+1)}),\dots\\
\qquad\qquad\qquad\qquad\qquad\qquad\qquad\qquad\qquad  \dots, Y_{t+l}^{(r)}=f_{s_{t+l}}^{\prime}(y_{t+l-1}^{(r+1)},y_{t+l}^{(r+1)})\}\\
\quad =P\{Y_{t+1}^{(r+1)}=y_{t+1}^{(r+1)}\}P\{Y_{t+2}^{(r)}=f_{s_{t+2}}^{\prime}(y_{t+1}^{(r+1)},y_{t+2}^{(r+1)}),\dots\\
\qquad\qquad\qquad\qquad\qquad\qquad\qquad\qquad\qquad  \dots, Y_{t+l}^{(r)}=f_{s_{t+l}}^{\prime}(y_{t+l-1}^{(r+1)},y_{t+l}^{(r+1)})\}.
\end{array}$$
The last equality is obtained by using the fact that $Y_{t+1}^{(r+1)}$ is independent of the
vector $(Y_{t+2}^{(r)},\dots,Y_{t+l}^{(r)})$, since $Y_{t+2}^{(r)},\dots,Y_{t+l}^{(r)}$ are not used for obtaining $Y_{t+1}^{(r+1)}$.

Using the
inductive hypothesis
$(Y_{t+2}^{(r)},\ldots,Y_{t+l}^{(r)})\sim U(\{1,2,\dots,
a\}^{l-1})$, $Y_{t+1}^{(r+1)}\sim U(\{1,2,\dots, a\})$ and from
previous expression we obtain that

$$P\{Y_{t+1}^{(r+1)}=y_{t+1}^{(r+1)},Y_{t+2}^{(r+1)}=y_{t+2}^{(r+1)},\dots,Y_{t+l}^{(r+1)}=y_{t+l}^{(r+1)}\}=\frac{1}{a}\cdot
\frac{1}{a^{l-1}}=\frac{1}{a^l}.$$

So, we have proved that $(Y_{t+1}^{(n)},Y_{t+2}^{(n)}\dots
Y_{t+l}^{(n)})\sim U(\{1,2,\dots, a\}^l)$ for each $l\leq n$ and
each $t\geq 0$.
\end{proof}


\section{Experimental results}

We made many experiments in order to present our theoretical results. Here we give an example. We have randomly chosen a message $M$ with
1,000,000 letters of the alphabet $A=\{1,2,3,4\}$ with the distribution of letters given in the Table \ref{input dist.} .

\begin{center}
\begin{table}[h]
\caption{The distribution of the letters in the input message}
\begin{center}
$\begin{tabular}{c c c c} \hline
1&2&3&4\\
\hline
\ 0.70\ \ &\ 0.15\ \ &\ 0.10\ \  &\ 0.05\ \ \\
\hline
\end{tabular}$
\end{center}
\label{input dist.}
\end{table}
\end{center}

We used the quasigroup  (\ref{kvazigrupa 40}) and its parastrophes.
\begin{equation}
\begin{tabular}{c|cccc}
\ \ $*$\ \  & \ \ $1$ \ \ &\ \  $2$ \ \ &\ \  $3$ \ \ &\ \ $4$\ \ \\
\hline
$1$ &  $1$ & $2$ & $4$ &\  $3$ \\
 $2$ & $3$ & $4$ & $2$ &\  $1$\\
 $3$ & $4$ & $3$ & $1$ &\  $2$\\
 $4$ & $2$ & $1$ & $3$ &\  $4$
\end{tabular}
\label{kvazigrupa 40}
\end{equation}

After applying $PE^{(3)}$ on $M$, we got a encrypted message $C=PE^{(3)}(M)$. In each $PE$-transformation, we chose the length of the first block $d_1=3$ and the initial leader $l_1=4$.

The distribution of letters in the output $C$ is given in the Table \ref{output dist.}.
\begin{center}
\begin{table}[h]
\caption{The distribution of the letters in the output message}
\begin{center}
$\begin{tabular}{c c c c } \hline
1&2&3&4\\
\hline
\ 0.2501\ \ &\ 0.2393\ \ &\ 0.2576\ \  &\ 0.2530\ \ \\
\hline
\end{tabular}$
\end{center}
\label{output dist.} %
\end{table}
\end{center}

It is obvious that the distribution of letters in the output message $C$ is uniform.

The distribution of pairs, triplets and 4-tuples of letters in $C$ are given on the Figure \ref{figure 1}, Figure \ref{figure 2} and Figure \ref{figure 3}. On the Figure \ref{figure 1}, the pairs are presented on the $x$-axis in the lexicographic order ($'11' \rightarrow 1$, $'12' \rightarrow 2$, \dots, $'44'\rightarrow 16$).  On the similar way, the triplets and 4-tuples are presented on Figure \ref{figure 2} and Figure \ref{figure 3}.

\begin{figure}[h!]
\begin{center}
\frame{\includegraphics[width=2.5in]{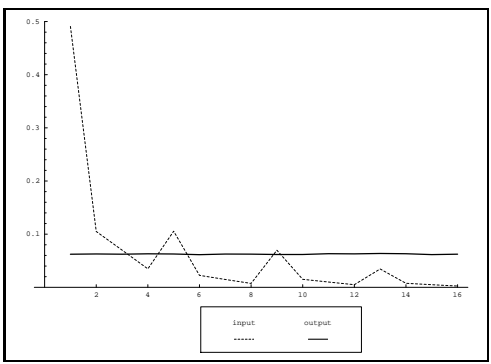}}
\caption{The distribution of the pairs in the input
message and the output message}
\label{figure 1}
\end{center}
\end{figure}

\begin{figure}[h!]
\begin{center}
\frame{\includegraphics[width=2.5in]{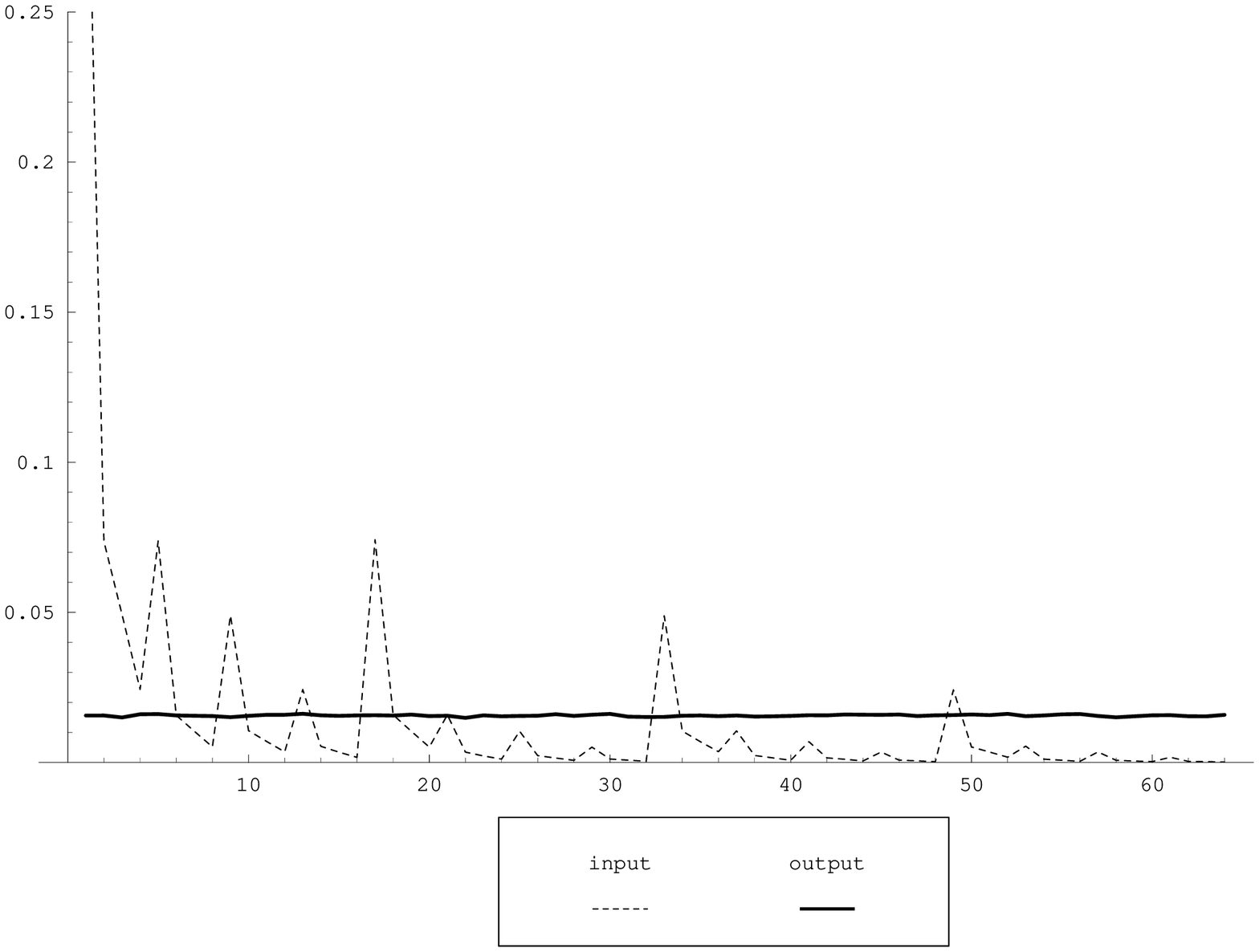}}
\caption{The distribution of the triplets in the input
message and the output message}
\label{figure 2}
\end{center}
\end{figure}

\begin{figure}[h!]
\begin{center}
\frame{\includegraphics[width=2.5in]{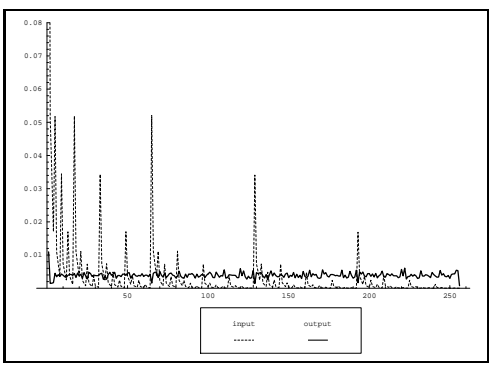}}
\caption{The distribution of the 4-tuples in the input
message and the output message}
\label{figure 3}
\end{center}
\end{figure}

We can see on Figure \ref{figure 1} and Figure \ref{figure 2} that after three applications of $PE$-transformations, the pairs and triplets are also uniformly
distributed as we proved in Theorem \ref{uniformity 2}.  Also, we can see on  Figure \ref{figure 3} that the distribution
of the 4-tuples in $C$ is not uniform, but that distribution is closer to the uniform distribution
than the distribution of 4-tuples  in the input message.

Next, we check whether Theorem 2 in \cite{VV} is satisfied when $PE$-transformation is applied. The distribution of pairs after one application of $PE$-transformation is presented on Figure \ref{figure 4}  a). On Figure \ref{figure 4}  b), we present the distribution of pairs after one application of $E$-transformation.  We can see that probabilities of pairs are divided in 4 classes on Figure \ref{figure 4} b) as the Theorem 2 in \cite{VV} claims. But we cannot distinguish any classes for probabilities on Figure  \ref{figure 4} a). This means that the algorithm for cryptanalysis proposed in \cite{VV} cannot be applied when an input message is encrypted by $PE$-transformation. Therefore encryption by $PE$-transformation is more resistant on statistical kind of attacks.

\begin{figure}[h!]
\begin{center}
\frame{\includegraphics[width=2.3in]{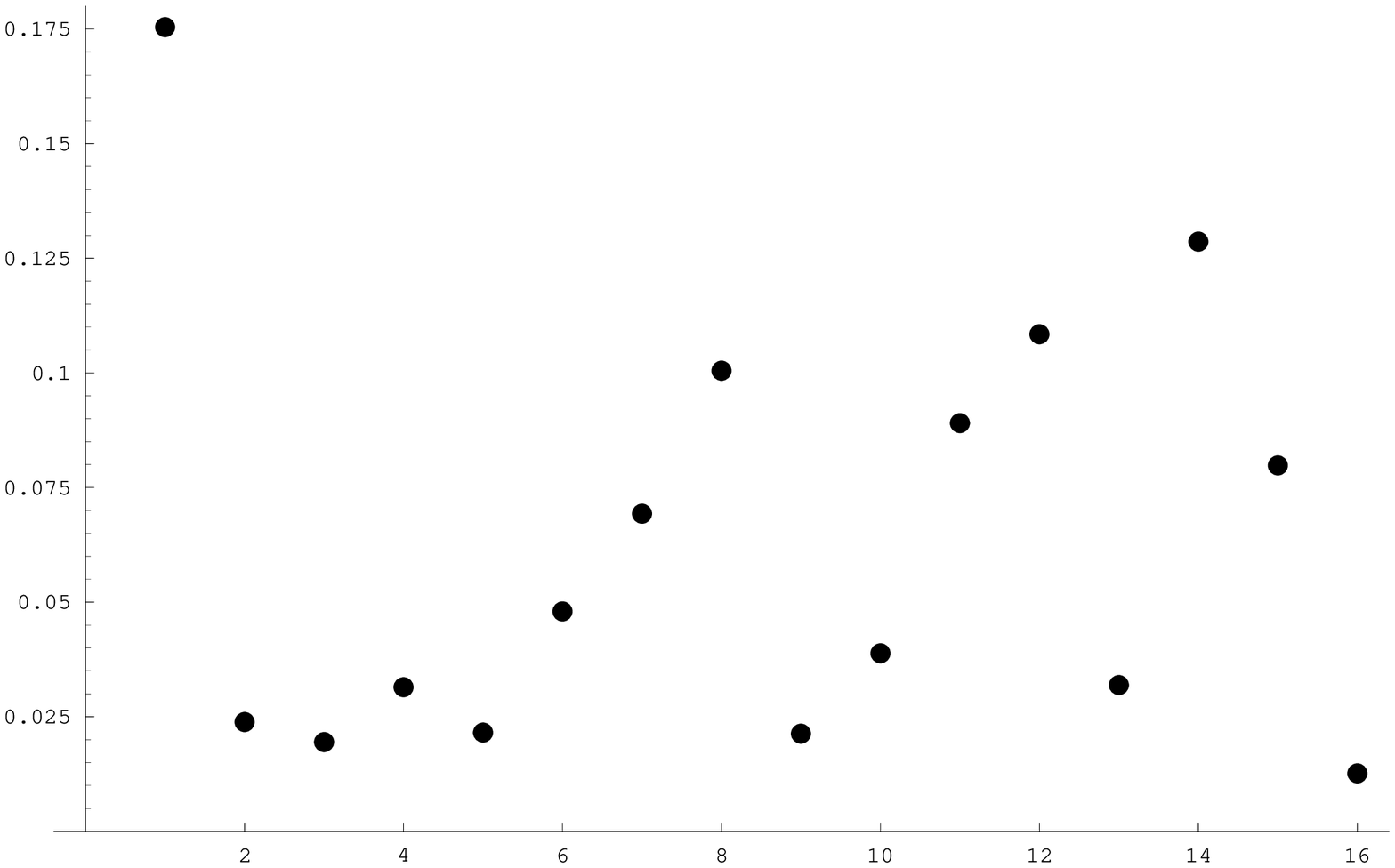}}
\quad
\frame{\includegraphics[width=2.3in]{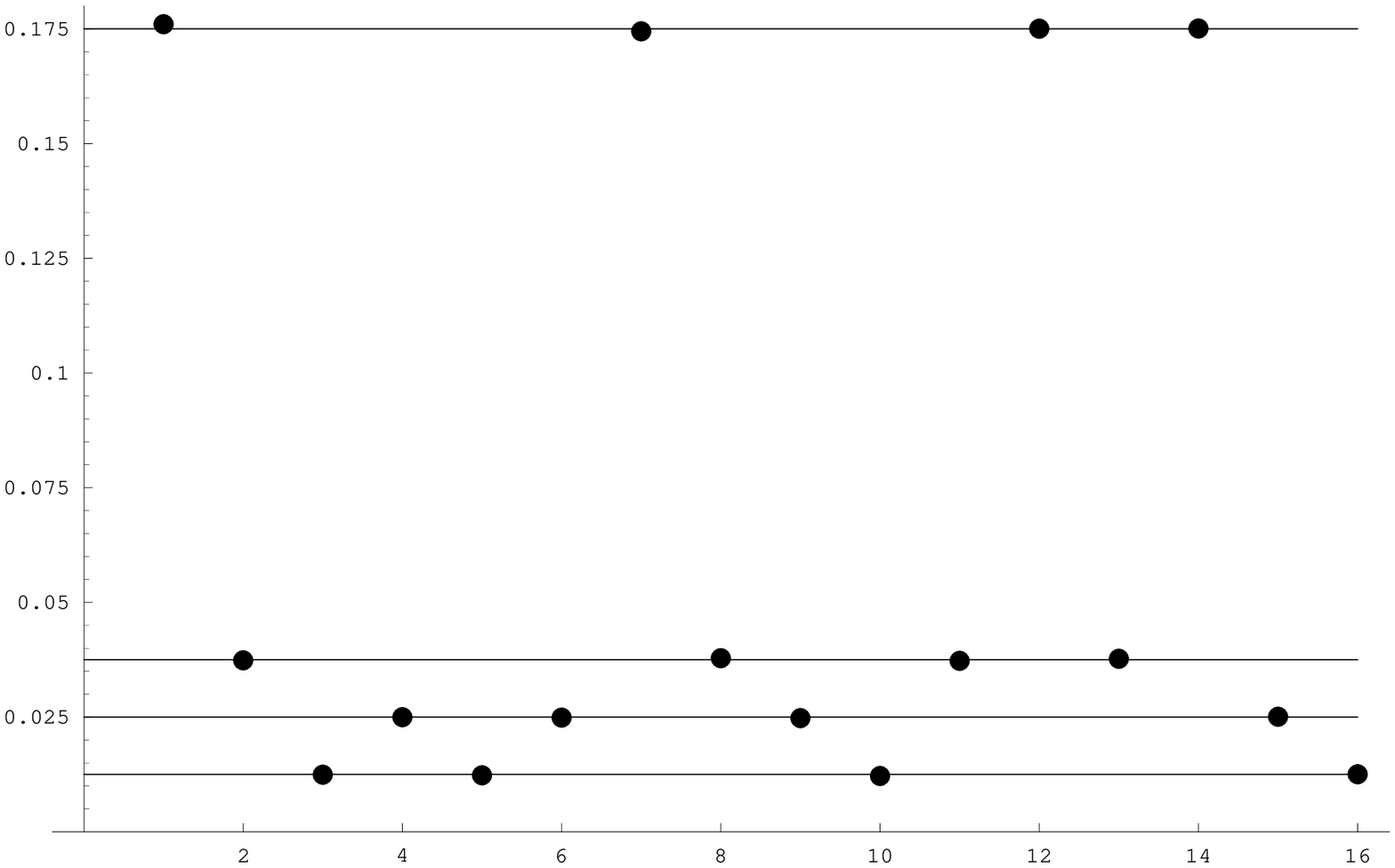}}\\
a)\quad\quad\quad\quad\quad\quad\quad\quad\quad\quad\quad\quad\quad\quad\quad\quad \quad\quad\quad b)
\caption{The distributions of the pairs in output messages obtained by $PE$- and $E$-transformation}
\label{figure 4}
\end{center}
\end{figure}

Note that for relevant statistical analyses, we must have enough large input message. Namely, in experiments, the probabilities of $n$-tuples are computed as relative frequencies. So, a relative frequency of an event tends to probability only if we have enough large sample.  The relevant statistical analyses cannot be done for shorter message. Therefore, statistical kind of attack is impossible on not enough large input message. Note that if an intruder catches and concatenates a lot of short messages
encrypted by the same $PE^{(n)}$-transformation, it will obtain a long message and it can apply a statistical attack. But, the attack will be impossible if we  change quasigroups used in encryption $PE^{(n)}$-transformation more often.


\section{Conclusion}

In this paper we proved that after $n$ applications of $PE$-transformation on an arbitrary message the distribution of $l$-tuples ($l=1,\dots, n$) is uniform and we cannot distinguish classes of probabilities in the distribution of $(n+1)$-tuples. This means that if $PE$-transformation is used as encryption function the obtained cipher messages are resistant on statistical kind of attacks when the number $n$ of applications of $PE$-transformation is enough large.

In \cite{SDV-1}, the authors concluded that $E$-transformation can be applied in cryptography as encryption function since  the number of quasigroups is huge one (there are more than $10^{58000}$ quasigroups when $|A|=256$) and the brute force attack is not reasonable.

If $PE$-transformation is used in encryption algorithm then the secret key will be a triplet  $(*, l, d_1)$. In that case, the brute force attack also is not possible since except the quasigroup operation $*$ and leader $l$, the key contains the length of the first block $d_1$ which has influence of the dynamic of changing of parastrophes.

At the end, in \cite{VVAA} authors proved that $PE$-transformation has better cryptographic properties than $E$-transformation for quasigroups of order 4.  Namely, some of fractal quasigroups of order 4 become parastrophic non-fractal and they can be used for designing of cryptographic primitives. Investigation for quasigroups of larger order cannot be done in real time since their number is very large.

Finally, from all results we can conclude that $PE$-transformation is better as encrypting function than $E$-transformation.


\end{document}